\documentclass[lettersize,journal]{IEEEtran}
\usepackage{amsmath,amssymb,amsfonts}
\usepackage{array}
\usepackage[caption=false,font=normalsize,labelfont=sf,textfont=sf]{subfig}
\usepackage{textcomp}
\usepackage{stfloats}
\usepackage{url}
\usepackage{verbatim}
\usepackage{graphicx}
\usepackage{cite}
\usepackage{graphicx,algorithm,algpseudocode,algorithmicx}
\usepackage{textcomp}
\usepackage{xcolor,cuted}
\usepackage{booktabs}
\usepackage{epstopdf,epsfig}
\usepackage{comment}
\usepackage[flushleft]{threeparttable}
\usepackage[font=small,labelfont=bf]{caption}
\usepackage{subfig}
\usepackage{tikz}
\usetikzlibrary{arrows,shapes,chains}
\newtheorem{lemma}{Lemma}
\newtheorem{theorem}{Theorem}

\newenvironment{proof}{\textbf{Proof:}}{\hfill$\square$}

\begin{document}

\title{A Leader-Follower Approach for The Attitude Synchronization of Multiple Rigid Body Systems on $SO(3)$}

\author{Yiliang Li, Jun-e Feng, and Abdelhamid Tayebi, \IEEEmembership{Fellow, IEEE}

\thanks{Yiliang Li and Jun-e Feng are with School of Mathematics, Shandong University, Jinan, Shandong 250100, P.R. China (e-mail: liyiliang1994@126.com,fengjune@sdu.edu.cn).}
\thanks{Abdelhamid Tayebi is with School of Mathematics, Shandong University, Jinan, Shandong 250100, P.R. China, Department of Electrical Engineering, Lakehead University, Thunder Bay, Ontario, Canada (e-mail: \{hlyu1,atayebi\}@lakeheadu.ca).}
\thanks{This work was supported by the National Sciences and Engineering Research Council of Canada (NSERC), under the grant NSERC-DG RGPIN 2020-06270, the National Natural Science Foundation of China, under the grant 62273201; 62350037, and the Research Fund for the Taishan Scholar Project of Shandong Province of China under the grant tstp20221103.}}



\maketitle

\begin{abstract}
This paper deals with the leader-follower attitude synchronization problem for a group of heterogeneous rigid body systems on $SO(3)$ under an undirected, connected, and acyclic graph communication topology. The proposed distributed control strategy, endowed with almost global asymptotic stability guarantees, allows the synchronization of the rigid body systems to a constant desired orientation known only to a single rigid body.
Some simulation results are also provided to validate the theoretical developments and illustrate the performance of the proposed control strategy.
\end{abstract}


\section{Introduction}

Attitude synchronization of a group of rigid body systems refers to the process of coordinating the orientation (\textit{i.e.,} rotational state) of multiple agents, such as drones, satellites, underwater vehicles, or robotic manipulators, so that they align with one another or track a common attitude reference. This problem is fundamental in multi-agent systems, and it has become increasingly important as modern applications rely on coordinated behavior among autonomous platforms \cite{abdessameud2013motion}.
The attitude synchronization problem for multiple rigid-body systems focuses on designing \emph{distributed control laws} that rely solely on local information,
typically obtained from neighboring agents in a communication graph. The objective is to ensure that the orientations of all rigid body systems evolve cooperatively and ultimately align with a common attitude. A major challenge in addressing this problem arises from the mathematical nature of the attitude representation. The only representation that describes rigid-body orientation \emph{globally}, \emph{uniquely}, and \emph{without singularities} is
the rotation matrix, which belongs to the special orthogonal group $SO(3)$. Due to the fact that $SO(3)$ is a compact matrix Lie group which is not homeomorphic to any Euclidean space, classical synchronization or consensus approaches developed for multi-agent systems evolving in $\mathbb{R}^n$ cannot be directly extended to this setting. In Euclidean spaces, consensus algorithms typically rely on linear operations that are not well defined on curved manifolds like $SO(3)$. Furthermore, the topology of $SO(3)$ prohibits the existence of any smooth feedback law that achieves global asymptotic stabilization to a single attitude \cite{Koditschek,Bhat_SCL2000}, rendering the synchronization problem intrinsically more complex. These geometric and topological restrictions motivate the development of specialized control strategies that explicitly account for the structure of $SO(3)$. As a result, attitude synchronization on $SO(3)$ requires fundamentally different techniques from classical consensus problems, often leveraging tools from differential geometry to achieve almost-global (or global) synchronization while respecting the manifold's intrinsic properties.

Leaderless synchronization strategies (leaderless consensus) for a group of rigid body systems on $SO(3)$ with local convergence properties have been proposed, for instance, in \cite{thunberg2016consensus,deng2021attitude}, and almost global consensus strategies have been developed, for instance, in \cite{Tron_TAC2012,Markdahl_TAC2020,bough_Berk_Tay_Arxiv2025}. The work in \cite{bough_Berk_Tay_Arxiv2025} has also been extended via hybrid feedback techniques to achieve global asymptotic synchronization on $SO(3)$.
The leader follower synchronization problem on $SO(3)$, where the rigid body systems are required to synchronize to a desired orientation known only to a few agents, has been addressed in a few papers in the literature with limited stability results. For instance, an observer-based leader-follower synchronization schemes on $SO(3)$ have been proposed in \cite{peng2018specified,zou2018velocity} with local convergence results.
The leader-follower attitude synchronization problem for rigid body systems has been mainly dealt with in the literature using the modified Rodrigues parameters and the unit-quaternion representation of the attitude to leverage the classical multi-agent techniques developed on Euclidian spaces.
Some techniques using the modified Rodrigues parameters have been proposed in \cite{meng2010distributed,ren2009distributed,zou2016distributed} and other techniques using the unit quaternion representation have been proposed in \cite{liang2024edge,he2022leader,lu2020leader,abdessameud2012attitude}, for the leader-follower synchronization. Most of the existing techniques in this field claim convergence results and not asymptotic stability.
Note that the modified Rodrigues parameters provide a local and singular representation of the attitude while the unit quaternion representation is global but non-unique. The only global, non-singular and unique representation of the attitude is the rotation matrix on $SO(3)$.

In this paper, we propose a leader-follower distributed control law for the attitude synchronization of a group of heterogeneous rigid body systems on $SO(3)$ to a constant desired orientation known only to a single rigid body. 
The proposed distributed control law, relying on local information exchange described by an undirected, connected, and acyclic graph topology, is endowed with almost global asymptotic stability guarantees\footnote{An equilibrium point is almost globally asymptotically stable if it is Lyapunov stable and attractive from all initial conditions except from a set of zero Lebesgue measure.}.

The remainder of this paper is organized as follows.
Section \ref{sec2} provides the notations, and some graph theory concepts used in this paper. The problem under consideration is formulated in Section \ref{sec3}. The distributed control law is designed in Section \ref{sec4}, along with the stability proof. Section \ref{sec5} provides some stimulation results that illustrate the performance of the proposed distributed control law. Some concluding remarks are provided in Section \ref{sec6}.

\section{Preliminaries}\label{sec2}

\subsection{Notation}
The sets of real numbers, $n$-dimensional vectors and $n$-by-$m$ matrices are denoted by $\mathbb{R},\mathbb{R}^{n}$ and $\mathbb{R}^{n\times m}$, respectively.
The Euclidean norm of $x\in\mathbb{R}^n$ is defined as $\|x\|=\sqrt{x^{\top}x}$. We define the set of unit vectors in $\mathbb{R}^n$ as $\mathbb{S}^{n-1}=\{x\in\mathbb{R}^{n}|\|x\|=1\}$.
For two matrices $A,B\in\mathbb{R}^{n\times m}$, the Euclidean inner product of $A$ and $B$ is defined as $\langle\langle A, B\rangle\rangle=\mathbf{tr}(A^{\top}B)$, and the Frobenius norm of $A$ is defined as $\|A\|_F=\sqrt{\mathbf{tr}(A^{\top}A)}$.
In this paper, we use $\mathbf{0}_n$ and $\mathbf{0}_{n\times m}$ to represent the $n$-dimensional null vector and the $n$-by-$m$ null matrix, respectively.
A null matrix with an appropriate dimensionality is denoted as $\mathbf{0}$.
An $n$-dimensional identity matrix is denoted as $I_n$.
The kronecker product is represented by $\otimes$.

The special orthogonal group is defined as $SO(3)=\{R\in\mathbb{R}^{3\times 3}|R^{\top}R=RR^{\top}=I_3,\mathrm{det}(R)=1\}$, where $\mathrm{det}(R)$ stands for the determinant of $R$.
The angle-axis representation of a rotation matrix $R$ is as
\begin{equation}\label{angleaixs}
\mathcal{R}(\theta,u)=I_3+\sin(\theta)u^{\times}+(1-\cos(\theta))(u^{\times})^2,
\end{equation}
where $\theta\in\mathbb{R}$ and $u\in\mathbb{S}^{2}$ are the rotation angle and the rotation axis, respectively.
The \textit{Lie\ algebra} of $SO(3)$ is denoted by $\mathfrak{so}(3)=\{\Omega\in\mathbb{R}^{3\times 3}|\Omega=-\Omega^{\top}\}$.
The symbol $\times$ is the vector cross-product on $\mathbb{R}^3$, leading to a map $(\cdot)^{\times}:\mathbb{R}^3\rightarrow\mathfrak{so}(3)$ satisfying $x\times y=x^{\times}y$ for $x,y\in\mathbb{R}^3$, where $x^{\times}$ is a $3$-by-$3$ skew-symmetric matrix, generated by $x=[x_1\ x_2\ x_3]^{\top}$, as follows:
\[
x^{\times}=\left[
\begin{array}{ccc}
0&-x_3&x_2\\
x_3&0&-x_1\\
-x_2&x_1&0\\
\end{array}
\right]_.
\]
We define the map $\mathbf{vex}:\mathfrak{so}(3)\rightarrow\mathbb{R}^3$ such that $(\mathbf{vex}(\Omega))^{\times}=\Omega$, $\forall \Omega\in\mathfrak{so}(3)$ and $\mathbf{vex}(x^{\times})=x$, $\forall x\in\mathbb{R}^3$. Let $\mathbb{P}_a : \mathbb{R}^{3\times 3} \rightarrow \mathfrak{so}(3)$ be the projection map on the Lie algebra $\mathfrak{so}(3)$ such that $\mathbb{P}_a(A):=(A-A^\top )/2$. We define the composition map $\psi:=\mathbf{vex}\circ\mathbb{P}_a$ such that $\psi(A)=\mathbf{vex}(\mathbb{P}_a(A))=\frac{1}{2}\mathbf{vex}(A-A^{\top})=\frac{1}{2}[a_{32}-a_{23}\ a_{13}-a_{31}\ a_{21}-a_{12}]^{\top}$, for $A=[a_{ij}]\in\mathbb{R}^{3\times 3}$.
One can also verify that
\begin{equation}\label{cross}
x^{\times}y^{
\times}=yx^{\top}-(x^{\top}y)I_3,
\end{equation}
where $x,y\in\mathbb{R}^3$.
The normalized attitude norm on $SO(3)$ is defined as $|R|_I=\frac{1}{2}\sqrt{\mathbf{tr}(I-R)}$.
We define the matrix $\mathbf{E}(A):=\frac{1}{2}(\mathbf{tr}(A)I_3-A^{\top}), A\in\mathbb{R}^{3\times 3}$.

\subsection{Graph Theory}

A graph is formally defined as $\mathcal{G}=(\mathcal{V},\mathcal{E})$, where $\mathcal{V}=\{1,2,\ldots,N\}$ is a set of vertices, $\mathcal{E}\subset\mathcal{V}\times\mathcal{V}$ is a set of edges of $\mathcal{G}$ with an edge between $i$ and $j$ being denoted as $(i,j)$.
The neighborhood $\mathcal{N}(i)$ of the vertex $i$ is defined as $\mathcal{N}(i)=\{j\in\mathcal{V}|(j,i)\in\mathcal{E}\}$.
A sequence of distinct vertices $(i_1,i_2,\ldots,i_m)$ is said to be a path of length $m$ in $\mathcal{G}$ if $(i_k,i_{k+1})\in\mathcal{E}$ holds for all $k=1,\ldots,m-1$.
A graph $\mathcal{G}$ is undirected if for any vertices $i,j\in\mathcal{V}$, $(i,j)\in\mathcal{E}$ implies $(j,i)\in\mathcal{E}$, otherwise, $\mathcal{G}$ is directed.
An undirected graph $\mathcal{G}$ is said to be connected if there exists a path between each pair of distinct vertices in $\mathcal{G}$.
An undirected graph $\mathcal{G}$ contains a cycle if there exists a path from one vertex to itself.
An undirected graph $\mathcal{G}$ is an undirected tree if any two vertices are connected by exactly one path (\textit{i.e.,}
an undirected tree is an undirected, connected, and acyclic graph).

\section{Problem formulation}\label{sec3}

Consider a group of $N$ heterogeneous\footnote{The rigid body systems do not have necessarily the same inertia matrix.} rigid body systems governed by the following rotational dynamics:
\begin{equation}\label{e3.1}
\left\{
\begin{array}{rcl}
\dot{R}_i&=&R_i\omega_i^{\times},\\
J_i\dot{\omega}_i&=&-\omega_i^\times J_i\omega_i+\tau_i,\\
\end{array}
\right.
\end{equation}
for $i=1,\ldots,N$, where $R_i\in SO(3)$ is the attitude of the $i$-th rigid body and $\omega_i\in\mathbb{R}^3$ the angular velocity of the $i$-th rigid body with respect to the inertial frame expressed in the body-attached frame, $J_i\in\mathbb{R}^{3\times 3}$ is the positive definite inertia matrix of the $i$-th rigid body, and $\tau_i\in\mathbb{R}^3$ is the control torque of the $i$-th rigid body.
The interconnection between the rigid body systems is represented by an undirected graph $\mathcal{G}=(\mathcal{V},\mathcal{E})$, where $\mathcal{V}=\{1,2,\ldots,N\}$ represents the set of rigid body systems, $\mathcal{E}\subset\mathcal{V}\times\mathcal{V}$ is the edge set with $(i,j)\in\mathcal{E}$ indicating that the $j$-th rigid body receives the information from the $i$-th rigid body. We assume that $\mathcal{G}$ is an undirected tree.

The objective of this paper is to design a decentralized control torque $\tau_i$  such that the attitude of the $i$-th rigid body is synchronized to the desired attitude $R_0$, \emph{i.e.}, $\lim_{t\rightarrow\infty}R_i=R_0$ and $\lim_{t\rightarrow\infty}\omega_i=\mathbf{0}_3$, for all $i=1,\ldots,N$. The desired attitude $R_0$ is assumed to be constant and available to only one rigid body, denoted by $1$. The desired attitude $R_0$ is represented by a virtual leader $0$. The rigid body systems $i=1,\ldots,N$ are regarded as followers. The relative attitude between the $1$-st rigid body and the desired orientation $R_0$ is denoted by $\bar{R}_{1,0}=R_0^{\top}R_1$, leading to $\dot{\bar{R}}_{1,0}=\bar{R}_{1,0}\omega_1^{\times}$.

\section{Main result}\label{sec4}

For an undirected graph $\mathcal{G}$, there exist two edges between the $i$-th rigid body and the $j$-th rigid body, \emph{i.e.}, $(i,j),(j,i)$.
It is worth mentioning that the convergence of $R_i^{\top}R_j$ implies the convergence of $R_j^{\top}R_i$, and vice versa.
This allows us to specify an arbitrary virtual direction for any two interconnected rigid body systems in the graph $\mathcal{G}$ for stability analysis (see, for instance, \cite{Arcak2008,Bough_Tay_2025}).
For the stability analysis that will be conducted later, we will arbitrarily remove one of the edges between any two interconnected rigid body systems, and denote the rest of the edge set  $\mathcal{E}$ as $\overline{\mathcal{E}}$, and obtain a new graph $\overline{\mathcal{G}}=\{\mathcal{V},\overline{\mathcal{E}}\}$ with a virtual orientation specified in the graph $\mathcal{G}$.
According to the obtained graph $\overline{\mathcal{G}}$, the edge $(j,i)\in\overline{\mathcal{E}}$ is denoted as the $k$-th edge in the graph $\mathcal{G}$. The relative attitude between the rigid body systems $i,j$, described by edge $k$, is defined as $\bar{R}_{k}=R_j^{\top}R_i$.
It follows that
\[\dot{\bar{R}}_{k}=\bar{R}_{k}\bar{\omega}_{k}^{\times},\]
where $\bar{\omega}_{k}=\omega_i-\bar{R}_{k}^{\top}\omega_j$.
It is evident that $|\overline{\mathcal{E}}|=N-1$ due to the fact that the graph $\mathcal{G}$ is an undirected tree.
Denoting $\bar{R}=\mathrm{diag}(\bar{R}_1,\ldots,\bar{R}_{N-1}),\overline{W}=\mathrm{diag}(\bar{\omega}_1^{\times},\ldots,\bar{\omega}_{N-1}^{\times})$, one has
\begin{equation}\label{e4.2}
\begin{array}{l}
\dot{\bar{R}}=\bar{R}\overline{W}.\\
\end{array}
\end{equation}
The angular velocity dynamics in \eqref{e3.1} can be rewritten in the following compact form:
\begin{equation}\label{ane4.2}
\begin{array}{l}
J\dot{\omega}=-WJ \omega+\tau,\\
\end{array}
\end{equation}
where $J=\mathrm{diag}(J_1,\ldots,J_N)$, $W=\mathrm{diag}(\omega_1^{\times},\ldots,\omega_N^{\times})$, $\omega=[\omega_1^{\top}\ \cdots\ \omega_N^{\top}]^{\top}$ and $\tau=[\tau_1^{\top}\ \cdots\ \tau_N^{\top}]^{\top}$.

Denote the set of edges in the graph $\overline{\mathcal{G}}$ as $\mathcal{M}=\{1,\ldots,N-1\}$.
Let us define the matrix
$L\in\mathbb{R}^{3N\times 3(N-1)}$ with each $3$ by $3$ block $L_{ik}$ being defined as
\begin{equation}\label{L_matrix}
L_{ik}=\left\{
\begin{array}{ll}
-\bar{R}_k,&
k\in\mathcal{M}_i^{+},\\
I_3,&k\in\mathcal{M}_i^{-},\\
\mathbf{0}_{3\times 3},&\mathrm{otherwise},
\end{array}
\right.
\end{equation}
where $\mathcal{M}_i^+\subset \mathcal{M}$ is the set of edges in the graph $\overline{\mathcal{G}}$ with $i$ being the head of the edges, $\mathcal{M}_i^- \subset \mathcal{M}$ is the set of edges in the graph $\overline{\mathcal{G}}$ with $i$ being the tail of the edges.

Now, we state the following lemma that will be used in the proof of our theorem.
\begin{lemma}\label{le1}
Consider the matrix $L\in\mathbb{R}^{3N\times 3(N-1)}$ associated to the interconnection graph $\mathcal{G}$, with the block $L_{ik}$ defined in \eqref{L_matrix}.
Then $L_2x=\mathbf{0}_{3(N-1)}$ implies $x=\mathbf{0}_{3(N-1)}$, where
$L_2$ is obtained from $L$ by removing the first three rows, \textit{i.e.},
\[
L_2=\left[
\begin{array}{cccc}
L_{21}&L_{22}&\cdots&L_{2,N-1}\\
L_{31}&L_{32}&\cdots&L_{3,N-1}\\
\vdots&\vdots&\ddots&\vdots\\
L_{N1}&L_{N2}&\cdots&L_{N,N-1}\\
\end{array}
\right]_.
\]
\end{lemma}

\begin{proof}
This result is proven by showing that $L_2$ is nonsingular.
Assume that $\mathrm{rank}(L_2)<3(N-1)$, equivalently, $\mathrm{rank}(L_2^{\top})<3(N-1)$.
Let $x=[x_2^{\top}\ \cdots\ x_{N}^{\top}]^{\top}\neq\mathbf{0}_{3(N-1)}$ be one of the solutions of $L_2^{\top}x=\mathbf{0}_{3(N-1)}$, where $x_i\in\mathbb{R}^3,i\in\{2,3,\ldots,N\}$.
For $i\in\mathcal{V}\setminus\{1\}$ and $j\in\mathcal{N}(i)\setminus(\mathcal{N}(1)\cup\{1\})$, there exists $k\in\{1,\ldots,N-1\}$ such that $L_{ik}\neq\mathbf{0}_{3\times 3},L_{jk}\neq\mathbf{0}_{3\times 3}$ and $L_{lk}=\mathbf{0}_{3\times 3}$ for $l\in\mathcal{V}\setminus\{1,i,j\}$.
It follows from $L_2^{\top}x=\mathbf{0}_{3(N-1)}$ that $L^{\top}_{ik}x_i+L^{\top}_{jk}x_j+\sum\limits_{l\in\mathcal{V}\setminus\{1,i,j\}}L^{\top}_{lk}x_l=L^{\top}_{ik}x_i+L^{\top}_{jk}x_j=\mathbf{0}_3$ for $i\in\mathcal{V}\setminus\{1\},j\in\mathcal{N}(i)\setminus(\mathcal{N}(1)\cup\{1\})$.
Since $L_{ik}$ and $L_{jk}$ are nonsingular, one has $x_j=-(L^{\top}_{jk})^{-1}L_{ik}^{\top}x_i$ for $i\in\mathcal{V}\setminus\{1\},j\in\mathcal{N}(i)\setminus(\mathcal{N}(1)\cup\{1\})$.
If $i\in\mathcal{N}(1)$, then there exists $k_0\in\{1,\ldots,N-1\}$, satisfying $k_0\neq k$, such that $L_{ik_0}\neq\mathbf{0}_{3\times 3}$ and $L_{lk_0}=\mathbf{0}_{3\times 3}$ for $l\in\mathcal{V}\setminus\{1,i\}$.
According to $L_2^{\top}x=\mathbf{0}_{3(N-1)}$, one has $L^{\top}_{ik_0}x_i+\sum\limits_{l\in\mathcal{V}\setminus\{1, i\}}L^{\top}_{lk_0}x_l=L^{\top}_{ik_0}x_i=\mathbf{0}_3$.
Due to the fact that $L_{ik}$ is nonsingular, it follows that $x_i=\mathbf{0}_3$, leading to $x_j=-(L^{\top}_{jk})^{-1}L_{ik}^{\top}x_i=\mathbf{0}_3$.
If $i\notin\mathcal{N}(i)$,
one can find a path between the $l_0$-th rigid body and the $i$-th rigid body since the interconnection graph $\mathcal{G}$ is an undirected tree, where $l_0\in\mathcal{N}(1)$.
It is clear that $x_{l_0}=\mathbf{0}_3$.
Denote the path between the $l_0$-th rigid body and the $i$-th rigid body as $(l_0,l_1,\ldots,l_s,i)$.
Let $(l_0,l_1),\ldots,(l_s,i)$ be described by edges $k_0,\ldots,k_{s}$ respectively.
It is clear that $x_i=(-1)^{s+1}(L^{\top}_{ik_s})^{-1}L^{\top}_{l_sk_s}\cdots(L_{l_1k_0}^{\top})^{-1}L_{l_0k_0}^{\top}x_{l_0}$
$=\mathbf{0}_3$, leading to $x_j=-(L^{\top}_{jk})^{-1}L_{ik}^{\top}x_i=\mathbf{0}_3$.
Hence, we obtain $x=\mathbf{0}_{3(N-1)}$, which is a contradiction.
Consequently, one has $\mathrm{rank}(L_2)=\mathrm{rank}(L_2^{\top})=3(N-1)$, allowing us to conclude that $L_2x=\mathbf{0}_{3(N-1)}$ implies $x=\mathbf{0}_{3(N-1)}$.

\end{proof}

We propose the following distributed control torque $\tau_i$:
\begin{equation}\label{control}
\tau_i=-k_{\bar{R}_0}\psi(A_{i,0}R_0^{\top}R_i)-k_{\bar{R}}\sum\limits_{j\in\mathcal{N}(i)}\psi(A_{i,j}R_j^{\top}R_i)-k_{\omega}\omega_i,
\end{equation}
for each rigid body $i\in\mathcal{V}$,
where
$k_{\bar{R}_0},k_{\bar{R}},k_{\omega}>0$, $A_{i,j}=A_{i,j}^{\top}>0$ with three distinct eigenvalues and  $A_{i,j}=A_{j,i}$, $A_{1,0}=A_{1,0}^{\top}>0$ with three distinct eigenvalues,  $A_{i,0}=\mathbf{0}_{3\times 3},\forall i\in\mathcal{V}\setminus\{1\}$.

As in \cite{bough_Berk_Tay_Arxiv2025}, using the facts $\psi(AR)=-\psi(R^{\top}A)$ and $\psi(AR)=R^{\top}\psi(RA)$, $A=A^{\top}\in\mathbb{R}^{3\times 3}, R\in SO(3)$, one has
\begin{equation}\label{transformation}
\begin{array}{l}
\sum\limits_{j\in\mathcal{N}(i)}\psi(A_{i,j}R_j^{\top}R_i)=\sum\limits_{j\in\mathcal{I}_i}\psi(A_{i,j}R_j^{\top}R_i)\\
\ \ \ \ \ \ \ \ \ \ \ \ \ \ \ \ \ \ \ \ \ \ \ \ \ \ \ \ \ +\sum\limits_{j\in\mathcal{O}_i}\psi(A_{i,j}R_j^{\top}R_i)\\
\ =\sum\limits_{j\in\mathcal{I}_i}\psi(A_{i,j}R_j^{\top}R_i)-\sum\limits_{j\in\mathcal{O}_i}\psi(R_i^{\top}R_jA_{i,j})\\
\ =\sum\limits_{j\in\mathcal{I}_i}\psi(A_{i,j}R_j^{\top}R_i)-\sum\limits_{j\in\mathcal{O}_i}R_i^{\top}R_j\psi(A_{i,j}R_i^{\top}R_j)\\
\ =\sum\limits_{k=1}^{N-1}L_{ik}\psi(A_{k}\bar{R}_k),\\
\end{array}
\end{equation}
where $\mathcal{I}_i=\{j\in \mathcal{N}(i)~|~(j,i) \in \overline{\mathcal{E}}\}$, $\mathcal{O}_i=\{j\in \mathcal{N}(i)~|~(i,j) \in \overline{\mathcal{E}}\}$ and $A_k=A_{i,j}$ for $i,j\in\mathcal{V}$ with the edge between $i$ and $j$ being described by edge $k$.
Based on this, the control $\tau$ in \eqref{control} can be written as 
\begin{equation}\label{controlcompact}
\tau=-k_{\tilde{R}_0}\Psi_0-k_{\bar{R}}L\bar{\Psi}-k_{\omega}\omega,\end{equation}
where $\Psi_0=[\psi^{\top}(A_{0}\bar{R}_{1,0})~ \mathbf{0}_{3(N-1)}^{\top}]^{\top}$, $\bar{\Psi}=[\psi^{\top}(A_{1}\bar{R}_1)\ \cdots$
$\psi^{\top}(A_{N-1}\bar{R}_{N-1})]^{\top}$, $A_0=A_{1,0}$.
The closed-loop system is given by:
\begin{equation}\label{closedloop}
\left\{
\begin{array}{rcl}
\dot{\bar{R}}_{1,0}&=&\bar{R}_{1,0}\omega_1,\\
\dot{\bar{R}}&=&\bar{R}\overline{W},\\
J\dot{\omega}&=&-WJ \omega-k_{\bar{R}_0}\Psi_0-k_{\bar{R}}L\bar{\Psi}-k_{\omega} \omega.\\
\end{array}
\right.
\end{equation}

Now, one can state our main result in the following theorem:
\begin{theorem}\label{th1}
Consider a network of $N$ rigid body systems, whose attitude dynamics are given in \eqref{e3.1}.
Let the control torque $\tau_i$ be given by \eqref{control}. Assume that the interconnection graph $\mathcal{G}$ is an undirected tree. Then the following statements hold.
\begin{enumerate}
\item
The trajectories of the closed-loop system \eqref{closedloop} converge to the following set of equilibra $\mathcal{C}_v=\Xi^d\cup \Xi^u$, where
$\Xi^d=\{(\bar{R}_{1,0}=I_{3},\bar{R}=I_{3(N-1)},\omega=\mathbf{0}_{3N})\}$ and
$\Xi^{u}=\{(\bar{R}_{1,0}=\bar{R}_{0}^{\ast},\bar{R}=\mathrm{diag}(\bar{R}_1^{\ast},\ldots,\bar{R}_{N-1}^{\ast}),\omega=\mathbf{0}_{3N})~|~\bar{R}_{p}^{\ast}=I_3,\forall p\in\mathcal{M}^{I},\bar{R}_q^{\ast}=\mathcal{R}(\pi,\bar{u}_q),\forall q\in\mathcal{M}^{\pi},
\bar{u}_q\in\mathcal{E}_{\bar{v}_q}^{\mathbb{R}}\}$
with $\mathcal{M}^{I}\cup\mathcal{M}^{\pi}=\mathcal{M}\cup\{0\}$, $\mathcal{M}^{\pi}\neq \emptyset$,  
$\mathcal{E}_{\bar{v}_i}^{\mathbb{R}}=\{\bar{v}_i\in\mathbb{R}^{3}|\exists \lambda, s.t., A_{i}\bar{v}_{i}=\lambda \bar{v}_i,\|\bar{v}_i\|=1\},i=0,1,\ldots,N-1$.

\item The undesired equilibrium points in $\Xi^u$ are unstable and the associated stable manifold has zero Lebesgue measure. The desired equilibrium $\Xi^d$ is almost globally asymptotically stable.
\item All the rigid body attitudes synchronize to the desired attitude $R_0$, \textit{i.e.},  $\lim_{t\rightarrow\infty}R_i(t)=R_0$ and $\lim_{t\rightarrow\infty}\omega_i(t)=\mathbf{0}_3$ for all $i\in\mathcal{V}$.
\end{enumerate}
\end{theorem}

\begin{proof}
Let us prove the first item.
Consider the following potential function:
\[\begin{array}{rcl}
V&=&\frac{1}{2}(k_{\bar{R}_0}\mathbf{tr}(A_{0}(I_{3}-\bar{R}_{1,0}))+\omega^{\top}J\omega\\
\ &\ &\ \ \ \ \ +k_{\bar{R}}\mathbf{tr}(\mathbf{A}(I_{3(N-1)}-\bar{R}))).
\end{array}\]
where $\mathbf{A}=\mathrm{diag}(A_1,\ldots,A_{N-1})$.
Using the facts that $\text{tr}\left(M[x]^\times\right)=\text{tr}\left(\mathbb{P}_a(M)[x]^\times\right)$, $\text{tr}\left([x]^\times [y]^\times\right)=-2x^{\top}y$, $\forall x, y\in \mathbb{R}^3$ and $\forall M \in \mathbb{R}^{3\times3}$, one can show that $\frac{d}{dt}\mathbf{tr}(A_{k}(I_3-\bar{R}_{k}))=2\bar{\omega}^{\top}_k\psi(A_{k}\bar{R}_{k})$, and $\frac{d}{dt}\mathbf{tr}(A_{0}(I_3-\bar{R}_{1,0}))=2\omega^{\top}_1\psi(A_{0}\bar{R}_{1,0})$. Therefore, the time-derivative of $V$, in view of \eqref{ane4.2}, is given by
\begin{equation}\label{dot_V}
\begin{array}{rcl}
\dot{V}&=&k_{\bar{R}_0}\omega_1^{\top}\psi(A_{0}\bar{R}_{1,0})+\omega^{\top}J\dot{\omega}+k_{\bar{R}}\sum\limits_{k=1}^{N-1}\bar{\omega}_{k}^{\top}\psi(A_k\bar{R}_k)\\
\ &=&k_{\bar{R}_0}\omega^{\top}\Psi_0+\omega^{\top}(-WJ\omega+\tau)+k_{\bar{R}}\bar{\omega}^{\top}\bar{\Psi}\\
\ &=&\omega^{\top}(k_{\bar{R}_0}\Psi_0+\tau)+k_{\bar{R}}\bar{\omega}^{\top}\bar{\Psi},\\
\end{array}
\end{equation}
where $\bar{\omega}=[\bar{\omega}_1^{\top}\ \cdots\ \bar{\omega}_{N-1}^{\top}]^{\top}$.
One can verify that $\overline{\omega}=L^{\top}\omega$, leading to $\dot{V}=\omega^{\top}(k_{\bar{R}_0}\Psi_0+\tau)+k_{\bar{R}}\omega^{\top}L\bar{\Psi}$.
Substituting \eqref{controlcompact} 
into $\dot{V}$, one has $\dot{V}=-k_{\omega}\|\omega\|^2\leq 0$.
Hence, the desired equilibrium point $\Xi^d$ is stable.

As per LaSalle's invariance principle, all the trajectories of the closed-loop system \eqref{closedloop} must converge to largest invariant set characterized by $\dot{V}=0$. From $\dot{V}=0$, one has $\omega=\mathbf{0}_{3N}$ which implies $\dot{\omega}=\mathbf{0}_{3N}$. Consequently, from \eqref{closedloop}, one has
\begin{equation}\label{psi}
k_{\bar{R}_0}\Psi_0+k_{\bar{R}}L\bar{\Psi}=\mathbf{0}_{3N}.
\end{equation}
Denote the first $3$ rows of $L$ as $L_1$, \textit{i.e.}, $L=[L_1^{\top}\ L_2^{\top}]^{\top}$.
One can see that
\begin{equation}\label{pside}
\begin{array}{l}
k_{\bar{R}_0}\psi(A_{0}\bar{R}_{1,0})+k_{\bar{R}}L_1\bar{\Psi}=\mathbf{0}_{3},\\
k_{\bar{R}}L_2\bar{\Psi}=\mathbf{0}_{3(N-1)}.\\
\end{array}
\end{equation}
In terms of Lemma \ref{le1}, one has $\bar{\Psi}=\mathbf{0}_{3(N-1)}$, resulting in $\psi(A_{0}\bar{R}_{1,0})=\mathbf{0}_3$, or equivalently, $A_{0}\bar{R}_{1,0}=\bar{R}_{1,0}A_{0}$.
It follows from $\bar{\Psi}=\mathbf{0}_{3(N-1)}$ that
$\psi(A_k\bar{R}_k)=\mathbf{0}_{3}$, or equivalently, $A_k\bar{R}_k=\bar{R}_kA_k$, where $k\in\{1,\ldots,N-1\}$.
Therefore, using, for instance, Lemma 2 in \cite{Mayhew2011}, one can show that $\bar{R}_k$ and $\bar{R}_{1,0}$ converges to $\{I_3\}\cup\{\mathcal{R}(\pi,v)|v\in\mathcal{E}_{\bar{v}_k}^{\mathbb{R}}\}$ and $\{I_3\}\cup\{\mathcal{R}(\pi,v)|v\in\mathcal{E}_{\bar{v}_0}^{\mathbb{R}}\}$ respectively, where $k\in\{1,\ldots,N-1\}$.
Consequently, the trajectories of the closed-loop system \eqref{closedloop} converge to $\mathcal{C}_v$.

Now we prove the second item of the theorem.
We start by using Chetaev's theorem \cite{khalil2002nonlinear} to show the unstability of the undesired equilibria $\Xi^u$.
Let us consider an undesired equilibrium point in $\Xi^{u}$
and define the following real-valued function:
\begin{equation}\label{bar_V}
\overline{V}=\sum\limits_{q\in\mathcal{M}^{\pi}}k_q(\mathbf{tr}(A_q)-\lambda_{\bar{u}_q}^{A_q})-V,
\end{equation}
where
\[
k_q=\left\{
\begin{array}{ll}
 k_{\bar{R}_0},   & q=0, \\
 k_{\bar{R}},    & \mathrm{otherwise}.
\end{array}
\right.
\]
The time-derivative of $\overline{V}$ is given by
\[
\dot{\overline{V}}=-\dot{V}=k_{\omega}\|\omega\|^2>0.
\]

To show that there exists a domain, around the undesired equilibrium point, where $\overline{V}$ is positive definite, we first  provide a new expression of $\overline{V}$ in terms of the angle-axis representation, 
and then show that the undesired equilibrium points in $\Xi^{u}$ are saddle points for $\overline{V}$ by verifying that the Hessian matrix of $\overline{V}$ is an indefinite matrix.

Let $\bar{R}_{1,0}=I_3+\sin(\theta_{1,0})\mu_{1,0}^{\times}+(1-\cos(\theta_{1,0}))(\mu_{1,0}^{\times})^2$ and $\bar{R}_k=I_3+\sin(\theta_{k})\mu_k^{\times}+(1-\cos(\theta_{k}))(\mu_k^{\times})^2$, where $k\in\{1,\ldots,N-1\}$, $\theta_{1,0},\theta_k\in\mathbb{R}$ are rotation angles of $\bar{R}_{1,0}$ and $\bar{R}_k$ respectively, $\mu_{1,0},\mu_k\in\mathbb{S}^{2}$ are rotation axes of $\bar{R}_{1,0}$ and $\bar{R}_k$ respectively.
As per Lemma 2.2.5 in \cite{BerkaneSoulaimanephdtheis},
it follows that
\begin{equation}\label{trunit}
\begin{array}{rcl}
\mathbf{tr}(A_{0}(I-\bar{R}_{1,0}))&=&2(1-\cos(\theta_{1,0}))\mu_{1,0}^{\top}\mathbf{E}(A_{0})\mu_{1,0},\\
\mathbf{tr}(A_k(I-\bar{R}_{k}))&=&2(1-\cos(\theta_k))\mu_k^{\top}\mathbf{E}(A_k)\mu_k,\\
\end{array}
\end{equation}
Substituting \eqref{trunit} into \eqref{bar_V}, one has
\[\begin{array}{rcl}
\overline{V}&=&\sum\limits_{q\in\mathcal{M}^{\pi}}k_{q}(\mathbf{tr}(A_q)-\lambda_{\bar{u}_q}^{A_q})-k_{\bar{R}}\mu^{\top}(\nu\otimes I_3)\bar{\Omega}\mu\\
\ &\ &-k_{\bar{R}_0}(1-\cos(\theta_{1,0}))\mu_{1,0}^{\top}\mathbf{E}(A_{1,0})\mu_{1,0}-\frac{1}{2}\omega^{\top}J\omega,\\
\end{array}\]
where $\mu=[\mu_1^{\top}\ \cdots\ \mu_{N-1}^{\top}]^{\top}$, $\bar{\Omega}=\mathrm{diag}(\mathbf{E}(A_1),\ldots,$
$\mathbf{E}(A_{N-1}))$, $\nu=\mathrm{diag}(1-\cos(\theta_1),\ldots,1-\cos(\theta_{N-1}))$.
Since $u_1\in\mathcal{E}_{v_1}^{\mathbb{R}}$ and $\bar{u}_k\in\mathcal{E}_{\bar{v}_k}^{\mathbb{R}}$ for $k\in\{1,\ldots,N-1\}$, one has
\[
\begin{array}{rcl}
\bar{u}_0^{\top}\mathbf{E}(A_{0})\bar{u}_0&=&\frac{1}{2}(\mathbf{tr}(A_{0})-\lambda_{\bar{u}_0}^{A_{0}}),\\
\bar{u}_k^{\top}\mathbf{E}(A_k)\bar{u}_k&=&\frac{1}{2}(\mathbf{tr}(A_{k})-\lambda_{\bar{u}_k}^{A_k}),\\
\end{array}
\]
where $k\in\{1,\ldots,N-1\}$.
It is clear that $\mathbf{tr}(A_p(I_3-\bar{R}_p^{\ast}))=0$ for $\forall p\in\mathcal{M}^I$.
It follows that $\overline{V}=0$ for $(\bar{R}_{1,0}=\bar{R}_{0}^{\ast},\bar{R}=\mathrm{diag}(\bar{R}_1^{\ast},\ldots,\bar{R}_{N-1}^{\ast}),\omega=\mathbf{0}_{3N})\in\Xi^{u}$.

One can obtain the Hessian matrix of $\overline{V}$ at $(\bar{R}_{1,0}=\bar{R}_{0}^{\ast},$
$\bar{R}=\mathrm{diag}(\bar{R}_1^{\ast},\ldots,\bar{R}_{N-1}^{\ast}),\omega=\mathbf{0}_{3N})\in\Xi^{u}$, as follows:
\[
\mathrm{Hess}(\overline{V})=\left[
\begin{array}{ccccc}
\mathbf{H}_{11}&\mathbf{0}&\mathbf{0}&\mathbf{0}&\mathbf{0}\\
\mathbf{0}&\mathbf{H}_{22}&\mathbf{0}&\mathbf{0}&\mathbf{0}\\
\mathbf{0}&\mathbf{0}&\mathbf{H}_{33}&\mathbf{0}&\mathbf{0}\\
\mathbf{0}&\mathbf{0}&\mathbf{0}&\mathbf{H}_{44}&\mathbf{0}\\
\mathbf{0}&\mathbf{0}&\mathbf{0}&\mathbf{0}&-\frac{1}{2}J\\
\end{array}
\right]
\]
where $\mathbf{H}_{11}=\frac{k_{\bar{R}_0}}{2}(\mathbf{tr}(A_0)-\lambda_{\bar{u}_0}^{A_{0}})$, $\mathbf{H}_{22}=-4k_{\bar{R}_0}\mathbf{E}(A_{0})$, $\mathbf{H}_{33}=\frac{k_{\bar{R}}}{2}\mathrm{diag}(\mathbf{tr}(A_1)-\lambda_{\bar{u}_1}^{A_1},\ldots,\mathbf{tr}(A_{N-1})-\lambda_{\bar{u}_{N-1}}^{A_{N-1}})$, $\mathbf{H}_{44}=-4k_{\bar{R}}\bar{\Omega}$.
It is clear that $\mathbf{H}_{11}>0$ and $\mathbf{H}_{33}=\mathbf{H}_{33}^{\top}>0$ since $A_{0}$ and $A_k,k\in\{1,\ldots,N-1\}$ are positive definite.
As shown in Lemma 2.2.2 in \cite{BerkaneSoulaimanephdtheis}, one has $\lambda_{v}^{\mathbf{E}(A_{0})}=\frac{1}{2}(\mathbf{tr}(A_{0})-\lambda_{v}^{A_{0}})>0$ and $\lambda_{v}^{\mathbf{E}(A_{k})}=\frac{1}{2}(\mathbf{tr}(A_{k})-\lambda_{v}^{A_{k}})>0$ for $k\in\{1,\ldots,N-1\}$.
Hence, $\mathbf{H}_{22}$ and $\mathbf{H}_{44}$ are negative definite.
It is clear that $-\frac{1}{2}J$ is also negative definite.
As a result, one can conclude that $\mathrm{Hess}(\overline{V})$ is an indefinite matrix. Consequently, the undesired equilibrium points in $\Xi^u$ are saddle points for $\overline{V}$, which implies that there exists a point $(\bar{R}'_{1,0},\bar{R}',\omega')\in\mathbb{B}_r$ such that $\overline{V}(\bar{R}'_{1,0},\bar{R}',\omega')>0$, there also exists another point $(\bar{R}''_{1,0},\bar{R}'',\omega'')\in\mathbb{B}_r$ such that $\overline{V}(\bar{R}''_{1,0},\bar{R}'',\omega'')<0$, where $\mathbb{B}_r=\{(\bar{R}_{1,0},\bar{R},\omega)~|~0< |\bar{R}_{1,0}^{\top}\bar{R}_0^{\ast}|_I+\sum_{k=1}^{N-1}|\bar{R}_k^{\top}\bar{R}_k^{\ast}|_{I}+\|\omega\|$
$<r\}$ with $r>0$.
Consequently, $\mathbb{U}=\{(\bar{R}_{1,0},\bar{R},\omega)\in\mathbb{B}_r|\overline{V}(\bar{R}_{1,0},\bar{R},\omega)>0\}$ is non-empty.

To sum up, the undesired equilibria in $\Xi^{u}$ are unstable by virtue of Chetaev's theorem.

Now, to show that the stable manifold associated to the undesired equilibria has zero Lebesgue measure, we linearize the closed loop system around the undesired equilibria and show that the Jacobian matrix does not have eigenvalues on the imaginary axis. This, with the fact that the undesired equilibria are unstable, allows us to conclude the proof of the claim by virtue of the stable manifold theorem \cite{Perko_book}.

Let $(\bar{R}_0^{\ast})^{\top}\bar{R}_{1,0}=\exp(\bar{r}_{1,0}^{\times})$ and
$(\bar{R}_k^{\ast})^{\top}\bar{R}_k=\exp(\bar{r}_k^{\times})$, where $\bar{r}_{1,0}\in\mathbb{R}^3$, $\bar{r}_k\in\mathbb{R}^3$, $k\in\{1,\ldots,N-1\}$.
It is clear that $\exp(\bar{r}_{1,0}^{\times})\approx I_3+\bar{r}_{1,0}^{\times}$ and $\exp(\bar{r}_{k}^{\times})\approx I_3+\bar{r}_k^{\times},k\in\mathcal{M}$, around the undesired equilibrium point $\Xi^{u}$, implying that $\bar{R}_{1,0}\approx\bar{R}_0^{\ast}(I_3+\bar{r}_{1,0}^{\times})$ and $\bar{R}\approx\bar{R}_k^{\ast}(I_3+\bar{r}_k^{\times})$ for $k\in\mathcal{M}$.
It follows that
\begin{equation}\label{psi10}
\begin{array}{rcl}
\psi(A_{0}\bar{R}_{1,0})&\approx&\psi(A_{0}\bar{R}_0^{\ast}(I_3+\bar{r}_{1,0}^{\times}))\\
\ &=&\frac{1}{2}\mathbf{vex}(A_{0}\bar{R}_0^{\ast}\bar{r}_{1.0}^{\times}+\bar{r}_{1,0}^{\times}(A_{0}\bar{R}_0^{\ast})^{\top})\\
\ &=&\frac{1}{2}\mathbf{vex}(((\mathbf{tr}(A_{0}\bar{R}_0^{\ast})I_3+(A_{0}\bar{R}_0^{\ast})^{\top})\bar{r}_{1,0})^{\times})\\
\ &=&\mathbf{E}(A_{0}\bar{R}_0^{\ast})\bar{r}_{1,0},\\
\end{array}
\end{equation}
and
\begin{equation}\label{psik}
\psi(A_k\bar{R}_k)\approx\mathbf{E}(A_k\bar{R}_k^{\ast})\bar{r}_k,
\end{equation}
for $k\in\{1,\ldots,N-1\}$, leading to the linearized angular velocity dynamics around the undesired equilibrium points in $\Xi^{u}$ as follows:
\begin{equation}\label{omegaapprox}
\begin{array}{rcl}
J_i\dot{\omega}_i&=&-k_{\bar{R}_0}\mathbf{E}(A_{i,0}\bar{R}_0^{\ast})\bar{r}_{i,0}\\
\ &\ &-k_{\bar{R}}\sum\limits_{k=1}^{N-1}\bar{L}_{ik}\mathbf{E}(A_{k}\bar{R}_k^{\ast})\bar{r}_k-k_{\omega}\omega_i,
\end{array}
\end{equation}
for $i\in\{1,\ldots,N\}$, where $\bar{r}_{i,0}=\mathbf{0}_{3}$ for $i=2,3,\ldots,N$, $\bar{L}\in\mathbb{R}^{3N\times 3(N-1)}$ with
each $3$ by $3$ block $\bar{L}_{ik}$ given by
\[
\bar{L}_{ik}=\left\{
\begin{array}{ll}
-\bar{R}_k^{\ast},&
k\in\mathcal{M}_i^+,\\
I_3,&
k\in\mathcal{M}_i^-,\\
\mathbf{0}_{3\times 3},&\mathrm{otherwise}.
\end{array}
\right.
\]
Denoting $\bar{r}_{0}=[\bar{r}_{1,0}^{\top}\ \mathbf{0}_{1\times 3(N-1)}]^{\top}$,
$\bar{r}=[\bar{r}_1^{\top}\ \cdots\ \bar{r}_{N-1}^{\top}]^{\top}$, 
$\Gamma_0=[\mathbf{E}^{\top}(A_{0}\bar{R}_0^{\ast})\ \mathbf{0}_{3\times 3(N-1)}]^{\top}$, $\Gamma=\mathrm{diag}(\mathbf{E}(A_1\bar{R}_1^{\ast}),\ldots,$
$\mathbf{E}(A_{N-1}\bar{R}_{N-1}^{\ast}))$, the linearized system around the undesired equilibrium points in $\Xi^{u}$ is given by
\begin{equation}\label{approxdynamic}
\begin{array}{rcl}
\dot{\bar{r}}_{1,0}&=&\omega_1,\\
\dot{\bar{r}}&=&\bar{L}^{\top}\omega,\\
J\dot{\omega}&=&-k_{\bar{R}_0}\Gamma_0\bar{r}_{1,0}-k_{\bar{R}}\bar{L}\Gamma\bar{r}-k_{\omega}\omega.\\
\end{array}
\end{equation}
Hence, the Jacobian matrix around the undesired equilibrium points in $\Xi^{u}$ is as follows:
\[
\mathbf{M}=\left[
\begin{array}{ccc}
\mathbf{0}_{3\times 3}&\mathbf{0}_{3\times 3(N-1)}&\mathbf{M}_{13}\\
\mathbf{0}_{3(N-1)\times 3}&\mathbf{0}_{3(N-1)\times 3(N-1)}&\mathbf{M}_{23}\\
J^{-1}\mathbf{M}_{31}&J^{-1}\mathbf{M}_{32}&J^{-1}\mathbf{M}_{33}\\
\end{array}
\right]
\]
where $\mathbf{M}_{13}=[I_3\ \mathbf{0}_{3\times 3(N-1)}]$, $\mathbf{M}_{23}=\bar{L}^{\top}$, $\mathbf{M}_{31}=-k_{\bar{R}_0}\Gamma_0$, $\mathbf{M}_{32}=-k_{\bar{R}}\bar{L}\Gamma$, $\mathbf{M}_{33}=-k_{\omega}I_{3N}$. 

Now we show that the Jacobian matrix $\mathbf{M}$ has no imaginary eigenvalues.
Assume that $0$ is an eigenvalue of $\mathbf{M}$, and the eigenvector associated to it is denoted by $z=[z_1^{\top}\ z_2^{\top}\ z_3^{\top}]^{\top}$, where $z_1\in\mathbb{R}^{3},z_2\in\mathbb{R}^{3(N-1)},z_3\in\mathbb{R}^{3N}$.
Then, $\mathbf{M}z=\mathbf{0}_{6N}$ implies that
\begin{equation}\label{eigenvalue0}
\begin{array}{l}
\mathbf{M}_{13}z_3=\mathbf{0}_3,\\
\mathbf{M}_{23}z_3=\mathbf{0}_{3(N-1)},\\
J^{-1}\sum\limits_{j=1}^3\mathbf{M}_{3j}z_j=\mathbf{0}_{3N}.
\end{array}
\end{equation}
Let $z_3=[z_{31}^{\top}\ \cdots z_{3N}^{\top}]^{\top}$, where $z_{3j}\in\mathbb{R}^3, j\in\{1,\ldots,N\}$.
It can be seen from the first equation in \eqref{eigenvalue0} that
$z_{31}=\mathbf{0}_3$.
As per the proof of Lemma 2 in \cite{boughellaba2023distributed}, there exists a function $Q:(SO(3))^{N-1}\rightarrow\mathbb{R}^{3N\times 3}$ such that $z_3=Q(\bar{R}_1,\ldots,\bar{R}_{N-1})z_{31}$, leading to $z_3=\mathbf{0}_{3N}$.
This implies that the last equation in \eqref{eigenvalue0} becomes $\mathbf{M}_{31}z_1+\mathbf{M}_{32}z_2=\mathbf{0}_{3N}$, which leads to
\[
\begin{array}{l}
-k_{\bar{R}_0}\mathbf{E}(A_{0}\bar{R}_0^{\ast})z_1-k_{\bar{R}}\bar{L}_1\Gamma z_2=\mathbf{0}_3,\\
-k_{\bar{R}}\bar{L}_2\Gamma z_2=\mathbf{0}_{3(N-1)},
\end{array}
\]
where $\bar{L}_1\in\mathbb{R}^{3\times 3(N-1)}$ is obtained from $\bar{L}$ by taking the first three rows,  $\bar{L}_2\in\mathbb{R}^{3(N-1)\times 3(N-1)}$ is obtained from $\bar{L}$ by removing the first three rows, \textit{i.e.}, $\bar{L}=[\bar{L}_1^{\top}\ \bar{L}_2^{\top}]^{\top}$.
Since $\bar{L}$ is obtained from $L$ by replacing $\bar{R}_k$ by $\mathcal{R}(\pi,\bar{u}_k)$, $\bar{L}_2$ is nonsingular according to Lemma \ref{le1}.
Since  $\mathbf{E}(A_k\bar{R}_k^{\ast})$ is nonsingular, it is clear that $\Gamma$ is nonsingular.
It follows that $-k_{\bar{R}}\bar{L}_2\Gamma z_2=\mathbf{0}_{3(N-1)}$ implies $z_2=\mathbf{0}_{3(N-1)}$.
This results in $z_1=\mathbf{0}_3$ due to the fact that $\mathbf{E}(A_{0}\bar{R}_0^{\ast})$ is nonsingular.
To sum up, one has $z=\mathbf{0}_{6N}$, leading to a contradiction.
Therefore, $0$ is not an eigenvalue of $\mathbf{M}$.

Assume that $\mathbf{M}$ has an imaginary eigenvalue $i\lambda$, and the eigenvector associated to it is denoted by $z=[z_1^{\top}\ z_2^{\top}\ z_3^{\top}]^{\top}$, where $i^{2}=-1,\lambda\in\mathbb{R}\setminus\{0\}$, $z_1\in\mathbb{R}^{3},z_2\in\mathbb{R}^{3(N-1)},z_3\in\mathbb{R}^{3N}$.
Then, one has $\mathbf{M}z=i\lambda z$, implying that
\[
\left\{
\begin{array}{l}
\mathbf{M}_{13}z_3=i\lambda z_1,\\
\mathbf{M}_{23}z_3=i\lambda z_2,\\
J^{-1}\sum\limits_{j=1}^{3}\mathbf{M}_{3j}z_j=i\lambda z_3.\\
\end{array}
\right.
\]
Due to the fact that $\lambda\neq 0$, one has $z_j=-\frac{i}{\lambda}\mathbf{M}_{j3}z_3$ for $j=1,2$, resulting in
\begin{equation}\label{A4}
J^{-1}\mathbf{M}_{33}z_3-i(\frac{1}{\lambda}J^{-1}\sum\limits_{j=1}^{2}\mathbf{M}_{3j}\mathbf{M}_{j3}z_3+\lambda z_3)=\mathbf{0}_{3N}.
\end{equation}
Multiplying \eqref{A4} by $z_3^{\top}J$ on the left, one obtains
\[
z_3^{\top}\mathbf{M}_{33}z_3-i(\frac{1}{\lambda}\sum\limits_{j=1}^{2}z_3^{\top}\mathbf{M}_{3j}\mathbf{M}_{j3}z_3+\lambda z_3^{\top}Jz_3)=0,
\]
meaning that $0=z_3^{\top}\mathbf{M}_{33}z_3=-k_{\omega}\|z_3\|^2$, further inferring that $z_3=\mathbf{0}_{3N}$.
Based on this, we conclude that $z=\mathbf{0}_{6N}$ owing to $z_j=-\frac{i}{\lambda}\mathbf{M}_{j3}z_3=\mathbf{0}$ for $j=1,2$.
This is a contradiction.
Hence, $\mathbf{M}$ has no imaginary eigenvalue.

Combined with the fact that the undesired equilibrium points in $\Xi^{u}$ are unstable, $\mathbf{M}$ must have at least one eigenvalue with a positive real part.
According to the stable manifold theorem, the stable manifold associated to the undesired equilibrium points in $\Xi^{u}$ has a zero Lebesgue measure.
Consequently, $\Xi^d$ is almost globally attractive. This with the fact that $\Xi^d$ is stable, allows us to conclude that $\Xi^d$ is almost globally asymptotically stable. 
Finally, we prove the last item of the theorem.
We have shown that $\lim_{t\rightarrow\infty}\omega_i=\mathbf{0}_3$ for $i\in\mathcal{V}$ and $\lim_{t\rightarrow\infty}R_0^{\top}R_1=\lim_{t\rightarrow\infty}\bar{R}_{1,0}=I_3$, leading to $\lim_{t\rightarrow\infty}R_1=R_0$.
Since the interconnection graph $\mathcal{G}$ is an  undirected tree, there exists a path between the $i$-th rigid body and the $1$-st rigid body, where $i\in\mathcal{V}\setminus\{1\}$.
Assume that the path between rigid body systems $i$ and $1$ is denoted by $(1,k_1,\ldots,k_{s},i)$.
Then one has
\[
\begin{array}{l}
\lim_{t\rightarrow\infty}R_0^{\top}R_i=\lim_{t\rightarrow\infty}R_0^{\top}R_1R_1^{\top}R_{k_1}\cdots R_{k_s}^{\top}R_i\\
=(\lim_{t\rightarrow\infty}R_0^{\top}R_1)(\lim_{t\rightarrow\infty}R_1^{\top}R_{k_1}) \cdots (\lim_{t\rightarrow\infty}R_{k_s}^{\top}R_i)=I_3,
\end{array}
\]
leading to $\lim_{t\rightarrow\infty}R_i=R_0$  for $i\in\mathcal{V}\setminus\{1\}$.

\end{proof}

\section{Simulation}\label{sec5}
Consider a network of 7 rigid body systems interconnected as per the graph shown in Figure \ref{fig1.1}. The constant desired attitude $R_0$ is available only to the rigid body system $1$.
The neighbor sets are given as $\mathcal{N}(1)=\{3\},\mathcal{N}(2)=\{7\}, \mathcal{N}(3)=\{1,5,6\},\mathcal{N}(4)=\{5\},\mathcal{N}(5)=\{3,4\}, \mathcal{N}(6)=\{3,7\},\mathcal{N}(7)=\{2,6\}$.

\begin{figure}[!htbp]
\centering
\tikzstyle{startstop} = [rectangle, rounded corners, minimum width = 0.8cm, minimum height=0.5cm,text centered, draw = black]
\tikzstyle{test}=[diamond,aspect=2,draw,thin]
\tikzstyle{point}=[coordinate,on grid,]
\tikzstyle{line} = [draw, -latex']
\begin{tikzpicture}
\node[startstop,draw=white](1){\includegraphics[width=0.05\textwidth]{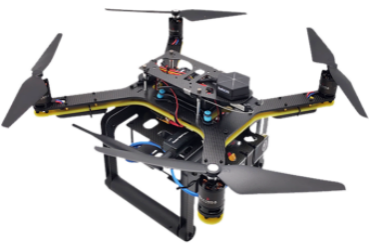}};
\node[startstop,right of=1,node distance=16mm,draw=red,fill=red](2){\includegraphics[width=0.05\textwidth]{uav.png}};
\node[startstop,right of=2,node distance=20mm,draw=white](3){\includegraphics[width=0.05\textwidth]{uav.png}};
\node[point,right of=3,node distance=12mm](point3){};
\node[startstop,below of=point3,node distance=12mm,draw=white](4){\includegraphics[width=0.05\textwidth]{uav.png}};
\node[point,left of=4,node distance=18mm](point4){};
\node[startstop,below of=point4,node distance=12mm,draw=white](5){\includegraphics[width=0.05\textwidth]{uav.png}};
\node[startstop,left of=5,node distance=20mm,draw=white](6){\includegraphics[width=0.05\textwidth]{uav.png}};
\node[startstop,left of=6,node distance=20mm,draw=white](7){\includegraphics[width=0.05\textwidth]{uav.png}};
\draw[-] (2)--(3);
\draw[-] (3)--(5);
\draw[-] (3)--(6);
\draw[-] (1)--(7);
\draw[-] (5)--(4);
\draw[-] (6)--(7);
\node[startstop,left of=1,node distance=9.8mm,draw=white](11){2};
\node[startstop,above of=2,node distance=7mm,draw=white](22){{\color{red}1}};
\node[startstop,above of=3,node distance=6mm,draw=white](33){3};
\node[startstop,above of=4,node distance=6mm,draw=white](44){4};
\node[startstop,right of=5,node distance=7mm,draw=white](55){5};
\node[startstop,above of=6,node distance=6mm,draw=white](66){6};
\node[startstop,left of=7,node distance=7mm,draw=white](77){7};
\end{tikzpicture}
\caption{An interconnection graph $\mathcal{G}$ of 7 rigid body systems}\label{fig1.1}
\end{figure}

The desired attitude is taken as $R_0=\mathcal{R}(0.8\pi,\frac{1}{\|u_0\|}u_0)$, where $u_0=[1\ 4\ 2]^{\top}$.
The inertia matrix $J_i$ of each rigid body is given by $J_i=\frac{1}{10}\mathrm{diag}(i,i+2,2i), i\in\{1,\ldots,7\}$. 
The initial conditions are as follows:
\[
\begin{array}{ll}
R_1(0)=\mathcal{R}(-0.1\pi,\frac{1}{\|u_2\|}u_2),&\omega_1(0)=[0\ 1\ 1]^{\top},\\
R_2(0)=\mathcal{R}(0.3\pi,u_1),&\omega_2(0)=[1\ 0\ 1]^{\top},\\
R_3(0)=\mathcal{R}(0.6\pi,u_1),&\omega_3(0)=[1\ 1\ 0]^{\top},\\
R_4(0)=\mathcal{R}(-0.2\pi,\frac{1}{\|u_2\|}u_2),&\omega_4(0)=[0\ 2\ 1]^{\top},\\
R_5(0)=\mathcal{R}(0.5\pi,u_1),&\omega_5(0)=[1\ 1\ 1]^{\top},\\
R_6(0)=\mathcal{R}(-0.2\pi,\frac{1}{\|u_2\|}u_2),&\omega_6(0)=[1\ 3\ 1]^{\top},\\
R_7(0)=\mathcal{R}(0.1\pi,\frac{1}{\|u_2\|}u_2),&\omega_7(0)=[3\ 0\ 1]^{\top},\\
\end{array}
\]
where $u_1=[0\ 1\ 0]^{\top},u_2=[1\ 1\ 0]^{\top}$.
The parameters are chosen as $k_{\bar{R}_0}=2.5,k_{\bar{R}}=2,k_{\tilde{\omega}}=1.5$ and
\[
\begin{array}{ll}
A_{1,0}=\mathrm{diag}(5,8,10),&\ \\
A_{2,7}=\mathrm{diag}(5,6,8),&A_{1,3}=\mathrm{diag}(6,8,10),\\
A_{3,5}=\mathrm{diag}(7,8,9),&A_{3,6}=\mathrm{diag}(5,7,8),\\
A_{4,5}=\mathrm{diag}(6,7,10),&A_{6,7}=\mathrm{diag}(5,7,10).\\
\end{array}
\]

The simulation results for the proposed control law \eqref{control} are shown in Fig. \ref{fig5.2} - Fig. \ref{fig5.5}.
Fig. \ref{fig5.2} - Fig. \ref{fig5.3} show that the attitude and angular velocity of each rigid body synchronize to the attitude and the angular velocity of its neighbors. Fig. \ref{fig5.4} - Fig. \ref{fig5.5} show that all the rigid body attitudes synchronize to the constant desired attitude $R_0$ and all the angular velocities converge to zero.

\begin{figure}[!htbp]
\centering
\includegraphics[width=0.5\textwidth]{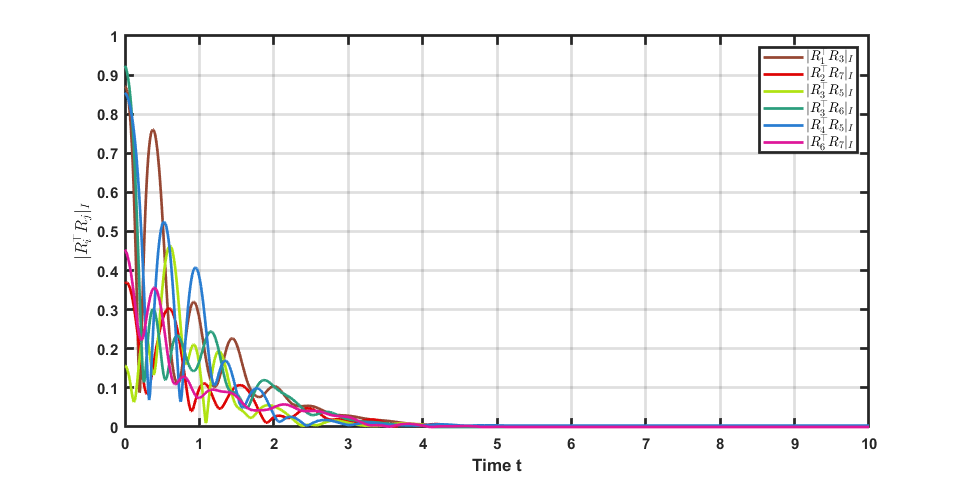}
\caption{The time evolution of the relative attitude associated with each edge}
\label{fig5.2}
\end{figure}

\begin{figure}[!htbp]
\centering
\includegraphics[width=0.5\textwidth]{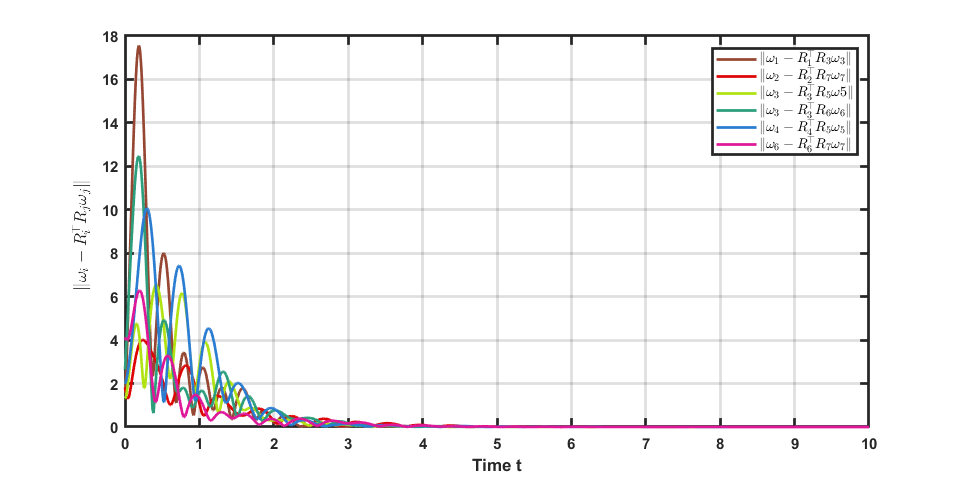}
\caption{The time evolution of the angular velocity error associated with each edge}
\label{fig5.3}
\end{figure}

\begin{figure}[!htbp]
\centering
\includegraphics[width=0.5\textwidth]{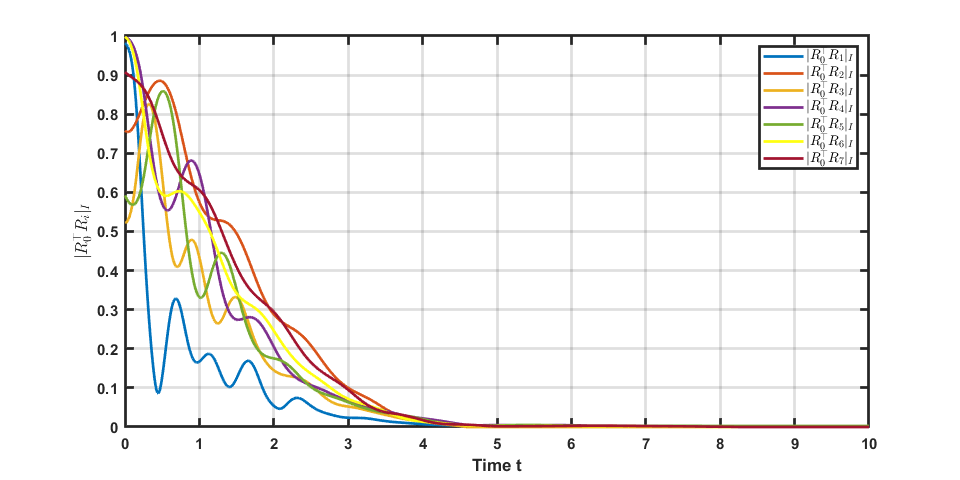}
\caption{The time evolution of the relative attitude between each rigid body system and the leader}
\label{fig5.4}
\end{figure}
\begin{figure}[!htbp]
\centering
\includegraphics[width=0.5\textwidth]{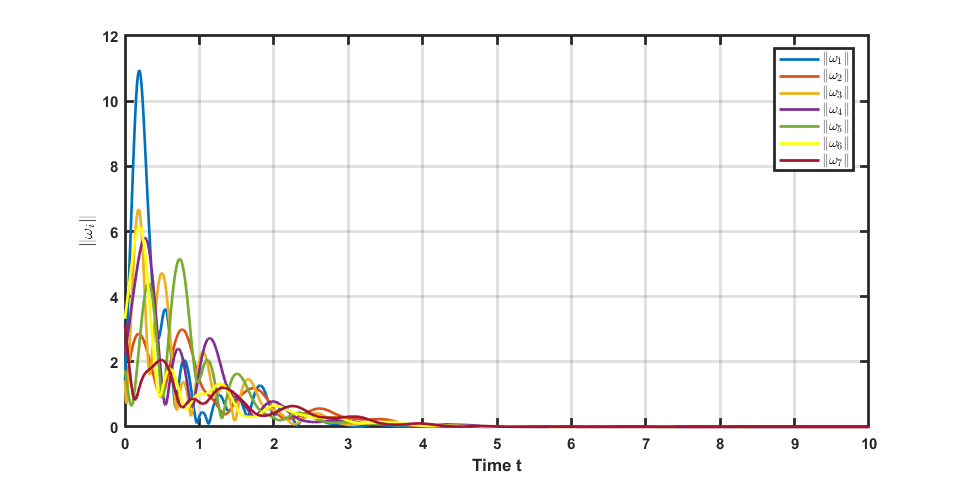}
\caption{The time evolution of the angular velocity associated with each rigid body system}
\label{fig5.5}
\end{figure}

\section{Conclusion}\label{sec6}

In this paper, we have proposed a leader-follower distributed control law guaranteeing almost global attitude synchronization of a group of heterogeneous rigid body systems on $SO(3)$. The proposed distributed control law relies solely on local information exchange described by an undirected, connected, and acyclic graph topology, and enables the attitude synchronization of the rigid body systems to a constant desired orientation known only to a single rigid body.
In our future research work, we will investigate a challenging open problem, namely, the leader-follower attitude synchronization problem on $SO(3)$, with almost global (or global) asymptotic stability guarantees, where the desired attitude is time-varying and the local information exchange between agents is subject to communication delays. We will also investigate the possibility of solving the problem under a more general interconnection graph topology.

\bibliographystyle{IEEEtran}
\bibliography{ifacconf}

\end{document}